\documentclass[letterpaper,10pt,twocolumn,twoside,journal]{IEEEtran}

\usepackage{times}
\usepackage{setspace}
\usepackage{url}
\spacing{1}
\usepackage[utf8]{inputenc}
\usepackage[T1]{fontenc}
\usepackage{graphicx}		
\usepackage{wrapfig}
\usepackage[format=plain,font=footnotesize,labelfont=bf,labelsep=period]{caption}
\usepackage{sidecap} 
\usepackage{subfig}
\usepackage[export]{adjustbox}
\usepackage[font=small]{caption}
\usepackage{float}

\usepackage{amsmath} 
\usepackage{amssymb}  
\usepackage{amsthm}
\usepackage{mathtools}
\usepackage{pdfsync}
\usepackage[normalem]{ulem}
\usepackage{paralist}	
\usepackage[space]{grffile} 
\usepackage{color}

\newtheorem{theorem}{Theorem}
\newtheorem{corollary}{Corollary}

\theoremstyle{definition}
\newtheorem{definition}{Definition}
\theoremstyle{remark}
\newtheorem{remark}{Remark}
\theoremstyle{definition}

\theoremstyle{definition}

\newcommand{\R}{\mathbb{R}}
\newcommand{\C}{\mathcal{C}}
\newcommand{\D}{\mathcal{D}}

\newcommand{\K}{\mathcal{K}}

\definecolor{darkblue}{RGB}{0,0,102}
\definecolor{lightblue}{RGB}{77,77,148}

\definecolor{gold}{RGB}{234, 170, 0}
\definecolor{metallic_gold}{RGB}{139, 111, 78}

\renewcommand{\cal}[1]{\mathcal{ #1 }}
\newcommand{\mb}[1]{\mathbf{ #1 }}
\newcommand{\bs}[1]{\boldsymbol{ #1 }}

\newcommand{\derp}[2]{\frac{\partial #1 }{\partial #2 }}

\DeclareMathOperator*{\esssup}{ess\,sup}

\usepackage{tikz-cd}

\title{\LARGE \textbf{A Control Barrier Perspective on Episodic Learning \\via Projection-to-State Safety}}
\author{Andrew J. Taylor, Andrew Singletary, Yisong Yue, Aaron D. Ames
\thanks{A. Taylor and Y. Yue are with the Department of Computing and Mathematical Sciences, California Institute of Technology, Pasadena, CA
    91125, USA, {\tt\small \{ajtaylor,yyue\}@caltech.edu}. A. Singletary and A. Ames are with the Department of Mechanical and Civil Engineering, California Institute of Technology, Pasadena, CA 91125, USA {\tt\small \{asinglet,ames\}@caltech.edu}.
}
}
\begin{document}

\maketitle

\begin{abstract}

In this paper we seek to quantify the ability of learning to improve safety guarantees endowed by Control Barrier Functions (CBFs). In particular, we investigate how model uncertainty in the time derivative of a CBF can be reduced via learning, and how this leads to stronger statements on the safe behavior of a system. To this end, we build upon the idea of Input-to-State Safety (ISSf) to define Projection-to-State Safety (PSSf), which characterizes degradation in safety in terms of a projected disturbance. This enables the direct quantification of both how learning can improve safety guarantees, and how bounds on learning error translate to bounds on degradation in safety. We demonstrate that a practical episodic learning approach can use PSSf to reduce uncertainty and improve safety guarantees in simulation and experimentally.

\end{abstract}

\section{Introduction}
\label{sec:intro}
Ensuring safety is of significant importance in the design of many modern control systems, from autonomous driving to industrial robotics. In practice, the models used in the control design process are imperfect, with model uncertainty arising due to parametric error and unmodeled dynamics. This uncertainty can cause the controller to render the system unsafe. As such,  it is necessary to quantify how the desired safety properties degrade with uncertainty.

Control Barrier Functions (CBFs) have become increasingly popular \cite{nguyen2016exponential, wang2017safety, ames2019control} as a tool for synthesizing controllers that provide safety via \textit{set invariance} \cite{blanchini2008set}. Safety guarantees endowed by a controller synthesized via CBFs rely on an accurate model of a system's dynamics, and may degrade in the presence of model uncertainty. The recently proposed definition of \textit{Input-to-State Safety} (ISSf) provides a tool for quantifying the impact on safety guarantees of such uncertainty or disturbances in the dynamics \cite{kolathaya2018input} by describing changes in the set kept invariant.

Due to its flexibility, it is increasingly popular to incorporate learning into safe controller synthesis \cite{wang2018safe,cheng2019end, ohnishi2019barrier,berkenkamp2016safe, fisac2018general}. 
Many of these approaches seek to provide statistical guarantees on the safety via assumptions made on learning performance. In practice however, limitations on learning performance arise due to factors such as covariate shift \cite{chen2016robust,liu2020robust}, limitations on model capacity, and optimization error. Thus, it is critical to understand the relationship between learning error and what safety guarantees can be ensured.

In this paper, we study how introducing learning models into safe controller synthesis done via CBFs can improve safety guarantees, and what safety guarantees can be made in the presence of learning error. In particular, we consider the episodic learning approach proposed in \cite{taylor2019learning}, where learning is done directly on the time derivative of a CBF. We integrate this approach with Input-to-State Safety to not only highlight how learning can intuitively lead to improved safety guarantees, but also provide a direct relationship between learning error and the degradation of safety guarantees. 

We make two main contributions in this paper. First, inspired by the idea of Projection-to-State Stability proposed in \cite{taylor2019control}, we  formulate general definitions of projections and projection compatible functions. Care must be taken to ensure these definitions preserve important topological properties for safety such as safe set membership. These definitions not only capture the definitions established in \cite{taylor2019control} as a special case, but allow us to define the notion of \textit{Projection-to-State Safety} (PSSf), which is a variant of the Input-to-State Safety property. Like ISSf, PSSf provides a tool for characterizing the degradation of safety in the presence of disturbances. Unlike ISSf, PSSf considers disturbances in a projected environment, allowing stronger guarantees on safe behavior. Second, we demonstrate the utility of PSSf by characterizing how data-driven learning models can improve safety guarantees, and how learning error leads to degradation in safety guarantees.

Our paper is organized as follows. Section \ref{sec:CBFs} provides a review of Control Barrier Functions and Input-to-State Safety. In Section \ref{sec:pssf} we define Projection-to-State Safety (PSSf) and discuss how PSSf enables quantifying degradation of safety in terms of a projected disturbance. Section \ref{sec:learning} defines a broad class of model uncertainty and explores how learning can be used to mitigate the impact of this uncertainty on safety. Lastly, in Section \ref{sec:simexp} we present both simulation and experimental results using PSSf to quantify the impact of learning error on safety guarantees for a Segway system.

\section{Preliminaries}
\label{sec:CBFs}
This section provides a review of Control Barrier Functions (CBFs) and Input-to-State Safe Control Barrier Functions (ISSf-CBFs). These tools will be used in Section \ref{sec:pssf} to define the notion of Projection-to-State Safety.

Consider the nonlinear control affine system given by:
\begin{equation}
    \label{eqn:eom}
    \dot{\mb{x}} = \mb{f}(\mb{x})+\mb{g}(\mb{x})\mb{u},
\end{equation}
where $\mb{x}\in\R^n$, $\mb{u}\in\R^m$, and $\mb{f}:\R^n\to\R^n$ and $\mb{g}:\R^n\to\R^{n\times m}$ are locally Lipschitz continuous on $\R^n$. Given a Lipschitz continuous state-feedback controller $\mb{k}:\R^n\to\R^m$, the closed-loop system dynamics are:
\begin{equation}
    \label{eqn:cloop}
    \dot{\mb{x}} = \mb{f}_{\textrm{cl}}(\mb{x}) \triangleq  \mb{f}(\mb{x})+\mb{g}(\mb{x})\mb{k}(\mb{x}).
\end{equation}
The assumption on local Lipschitz continuity of $\mb{f}$ and $\mb{k}$ implies that $\mb{f}_\textrm{cl}$ is locally Lipschitz continuous. Thus for any initial condition $\mb{x}_0 := \mb{x}(0) \in \R^n$ there exists a maximum time interval $I(\mb{x}_0) = [0, t_{\textrm{max}})$ such that $\mb{x}(t)$ is the unique solution to \eqref{eqn:cloop} on $I(\mb{x}_0)$ \cite{perko2013differential}. In the case that $\mb{f}_{\textrm{cl}}$ is forward complete, $t_{\textrm{max}}=\infty$.

A continuous function $\alpha:[0,a)\to\R_+$, with $a>0$, is said to belong to \textit{class $\cal{K}$} ($\alpha\in\cal{K}$) if $\alpha(0)=0$ and $\alpha$ is strictly monotonically increasing. If $a=\infty$ and $\lim_{r\to\infty}\alpha(r)=\infty$, then $\alpha$ is said to belong to \textit{class $\cal{K}_\infty$} ($\alpha\in\cal{K}_\infty$). A continuous function $\alpha:(-b,a)\to\R$, with $a,b>0$, is said to belong to \textit{extended class $\cal{K}$} ($\alpha\in\cal{K}_e$) if $\alpha(0)=0$ and $\alpha$ is strictly monotonically increasing. If $a,b=\infty$, $\lim_{r\to\infty}\alpha(r)=\infty$, and $\lim_{r\to-\infty}\alpha(r)=-\infty$, then $\alpha$ is said to belong to \textit{extended class $\cal{K}_\infty$} ($\alpha\in\cal{K}_{\infty,e}$)

The notion of safety that we consider is formalized by specifying a \textit{safe set} in the state space that the system must remain in to be considered safe. In particular, consider a set $\C\subset \R^n$ defined as the 0-superlevel set of a continuously differentiable function $h:\R^n \to \R$, yielding:
\begin{align}
    \C &\triangleq \left\{\mb{x} \in \R^n : h(\mb{x}) \geq 0\right\}, \label{eqn:safeset}\\
    \partial\C &\triangleq \{\mb{x} \in \R^n : h(\mb{x}) = 0\},\label{eqn:safesetboundary}\\
    \textrm{Int}(\C) &\triangleq \{\mb{x} \in \R^n : h(\mb{x}) > 0\}.\label{eqn:safetsetinterior}
\end{align}
We assume that $\C$ is nonempty and has no isolated points, that is, $\textrm{Int}(\C) \not = \emptyset \textrm{ and }\overline{\textrm{Int}(\C)} = \C$. We refer to $\cal{C}$ as the \textit{safe set}. This construction motivates the following definitions of forward invariant and safety:

\begin{definition}[\textit{Forward Invariant \& Safety}]
A set $\C\subset\R^n$ is \textit{forward invariant} if for every $\mb{x}_0\in\C$, the solution $\mb{x}(t)$ to \eqref{eqn:cloop} satisfies $\mb{x}(t) \in \C$ for all $t \in I(\mb{x}_0)$. The system \eqref{eqn:cloop} is \textit{safe} on the set $\C$ if the set $\C$ is forward invariant.
\end{definition}

Certifying the safety of the closed-loop system \eqref{eqn:cloop} with respect to a set $\C$ may be impossible if the controller $\mb{k}$ was not chosen to enforce the safety of $\C$. Control Barrier Functions can serve as a synthesis tool for attaining the forward invariance, and thus the safety of a set:
\begin{definition}[\textit{Control Barrier Function (CBF)}, \cite{ames2017control}]
Let $\C\subset\R^n$ be the 0-superlevel set of a continuously differentiable function $h:\R^n\to\R$ with $0$ a regular value. The function $h$ is a \textit{Control Barrier Function} (CBF) for \eqref{eqn:eom} on $\C$ if there exists $\alpha\in\K_{\infty,e}$ such that for all $\mb{x}\in\R^n$:
\begin{equation}
\label{eqn:cbf}
     \sup_{\mb{u}\in\R^m} \dot{h}(\mb{x},\mb{u}) \triangleq \derp{h}{\mb{x}}(\mb{x})\left(\mb{f}(\mb{x})+\mb{g}(\mb{x})\mb{u}\right)\geq-\alpha(h(\mb{x})).
\end{equation}
\end{definition}
We note that this definition can be relaxed such that the inequality only holds for all $\mb{x}\in E$ where $E$ is an open set satisfying $\C\subset E\subset\R^n$. Given a CBF $h$ for \eqref{eqn:eom} and a corresponding $\alpha\in\cal{K}_{\infty,e}$, we can consider the point-wise set of all control values that satisfy \eqref{eqn:cbf}:
\begin{equation*}
    K_{\textrm{cbf}}(\mb{x}) \triangleq \left\{\mb{u}\in\R^m ~\left|~ \dot{h}(\mb{x},\mb{u})\geq-\alpha(h(\mb{x})) \right. \right\}.
\end{equation*}
One of the main results in \cite{ames2014control, xu2015robustness} relates controllers taking values in  $K_{\textrm{cbf}}(\mb{x})$ to the safety of \eqref{eqn:eom} on $\C$:
\begin{theorem}
Given a set $\C\subset\R^n$ defined as the 0-superlevel set of a continuously differentiable function $h:\R^n\to\R$, if $h$ is a CBF for \eqref{eqn:eom} on $\C$, then any Lipschitz continuous controller $\mb{k}:\R^n\to\R^m$, such that $\mb{k}(\mb{x})\in K_{\textrm{cbf}}(\mb{x})$ for all $\mb{x}\in\R^n$, renders the system \eqref{eqn:eom} safe with respect to the set $\C$.
\end{theorem}

To accommodate disturbances or model uncertainties, we consider a disturbance space $\cal{D}\in\R^n$ and a disturbed system:
\begin{equation}
\label{eqn:eomdist}
    \dot{\mb{x}} = \mb{f}(\mb{x}) + \mb{g}(\mb{x})\mb{u}+\mb{d}.
\end{equation}
with $\mb{d}\in\D$. The disturbance may be time-varying, state and/or input dependent. We will assume that when viewing $\mb{d}$ as a signal, $\mb{d}(t)$, it is essentially bounded in time, and define $\Vert\mb{d}\Vert_\infty\triangleq\esssup_{t\geq 0}\Vert\mb{d}(t)\Vert$. Under a Lipschitz continuous state-feedback controller $\mb{k}$, the closed-loop dynamics are then given by:
\begin{equation}
\label{eqn:cloopdist}
    \dot{\mb{x}} = \mb{f}_{\textrm{cl}}(\mb{x},\mb{d}) \triangleq \mb{f}(\mb{x})+\mb{g}(\mb{x})\mb{k}(\mb{x})+\mb{d}.
\end{equation}
In the presence of disturbances, a controller $\mb{k}$ synthesized to render the set $\C$ safe for the undisturbed dynamics \eqref{eqn:cloop} may fail to render $\C$ safe for the disturbed dynamics \eqref{eqn:cloopdist}. To quantify how safety degrades, we consider the notion of \textit{input-to-state safety} \cite{kolathaya2018input}:

\begin{definition}[\textit{Input-to-State Safety (ISSf)}]
\label{def:issf}
The closed-loop system \eqref{eqn:cloopdist} is \textit{input-to-state safe} (ISSf) on a set $\C\subset\R^n$ with respect to disturbances $\mb{d}$ if there exists $\overline{d}>0$ and $\gamma\in\cal{K}_\infty$ such that the set $\C_\mb{d}\supset\C$ defined as:
\begin{align}
    \C_{\mb{d}} &\triangleq \left\{\mb{x} \in \R^n : h(\mb{x})+\gamma(\Vert\mb{d}\Vert_\infty) \geq 0\right\}, \label{eqn:safesetdist}\\
    \partial\C_{\mb{d}} &\triangleq \{\mb{x} \in \R^n : h(\mb{x})  +\gamma(\Vert\mb{d}\Vert_\infty) = 0\},\label{eqn:safesetboundarydist}\\
    \textrm{Int}(\C_{\mb{d}}) &\triangleq \{\mb{x} \in \R^n : h(\mb{x})+\gamma(\Vert\mb{d}\Vert_\infty) > 0\},\label{eqn:safetsetinteriordist}
\end{align}
is forward invariant for all $\mb{d}$ satisfying $\Vert\mb{d}\Vert_\infty\leq\overline{d}$.
\end{definition}
We refer to $\C$ as an \textit{input-to-state safe} set (ISSf set) if such a set $\C_{\mb{d}}$ exists. This definition implies that though the set $\C$ may not be safe, a larger set $\C_\mb{d}$, depending on $\mb{d}$, is safe. If $\mb{d}\equiv\mb{0}$, we recover that the set $\C$ is safe. $\C$ can be certified as an ISSf set for the closed-loop system \eqref{eqn:cloopdist} with the following definition:

\begin{definition}[\textit{Input-to-State Safe Barrier Function (ISSf-BF)}]
Let $\C\subset\R^n$ be the 0-superlevel set of a continuously differentiable function $h:\R^n\to\R$ with $0$ a regular value. The function $h:\R^n\to\R$ is an \textit{Input-to-State Safe Barrier Function} (ISSf-BF) for \eqref{eqn:cloopdist} on $\cal{C}$ if there exists $\overline{d}>0$, $\alpha\in\K_{\infty,e}$, and $\iota\in\cal{K}_\infty$ such that: 
\begin{equation}
\label{eqn:ISS-BF}
    \derp{h}{\mb{x}}(\mb{x})(\mb{f}(\mb{x})+\mb{g}(\mb{x})\mb{k}(\mb{x}) +\mb{d}) \geq -\alpha(h(\mb{x}))-\iota(\Vert\mb{d}\Vert),
\end{equation}
for all $\mb{x}\in\R^n$ and $\mb{d}\in\R^n$ such that $\Vert\mb{d}\Vert\leq\overline{d}$.
\end{definition}
As shown in \cite{kolathaya2018input}, the existence of an ISSf-BF for \eqref{eqn:cloopdist} on $\C$ implies $\C$ is an ISSf set. Similarly to the undisturbed case, we can introduce the notion of a Control Barrier Function for synthesizing controllers that ensure input-to-state safety:

\begin{definition}[\textit{ISSf Control Barrier Function (ISSf-CBF)}]
Let $\C\subset\R^n$ be the 0-superlevel set of a continuously differentiable function $h:\R^n\to\R$ with $0$ a regular value. The function $h$ is an \textit{Input-to-State Safe Control Barrier Function} (ISSf-CBF) for \eqref{eqn:eomdist} on $\C$ if there exists $\overline{d}>0$, $\alpha\in\K_{\infty,e}$, and $\iota\in\cal{K}_\infty$ such that: 
\begin{align}
\label{eqn:hdotbound}
   \sup_{\mb{u}\in\R^m} \dot{h}(\mb{x},\mb{u},\mb{d}) \triangleq \derp{h}{\mb{x}}(\mb{x})(&\mb{f}(\mb{x})+\mb{g}(\mb{x})\mb{u}+\mb{d})\nonumber\\&\geq -\alpha(h(\mb{x}))-\iota(\Vert\mb{d}\Vert),
\end{align}
for all $\mb{x}\in\R^n$ and $\mb{d}\in\R^n$ satisfying $\Vert\mb{d}\Vert\leq\overline{d}$.
\end{definition}
We note that this definition is a more general definition of an ISSf-CBF compared to \cite{kolathaya2018input}, where disturbances enter the system with the inputs. We define the pointwise set:
\begin{equation*}
    K_{\textrm{issf}}(\mb{x}) \triangleq \left\{\mb{u}\in\R^m ~\left|~ \dot{h}(\mb{x},\mb{u},\mb{d})\geq-\alpha(h(\mb{x}))-\iota(\Vert\mb{d}\Vert) \right. \right\},
\end{equation*}
noting that for a fixed input the inequality must hold for all $\mb{d}\in\R^n$ satisfying $\Vert\mb{d}\Vert\leq\overline{d}$. Given this result, we have the following theorem:
\begin{theorem}
\label{thm:issfcbf2issf}
Given a set $\C\subset\R^n$ defined as the 0-superlevel set of a continuously differentiable function $h:\R^n\to\R$, if $h$ is an ISSf-CBF for \eqref{eqn:eomdist} on $\C$, then any Lipschitz continuous controller $\mb{k}:\R^n\to\R^m$, such that $\mb{k}(\mb{x})\in K_{\textrm{issf}}(\mb{x})$ for all $\mb{x}\in\R^n$, renders the set $\C$ ISSf for \eqref{eqn:cloopdist}.
\end{theorem}
This theorem follows from the fact that under the controller $\mb{k}$, $h$ serves an ISSf-BF for \eqref{eqn:cloopdist} on $\C$.

\section{Projection-to-State Safety}
\label{sec:pssf}
Input-to-State Safety describes how the safe set $\C$ changes in terms of the disturbance as it appears in the state dynamics (see Definition \ref{def:issf} in Section \ref{sec:CBFs}). This description does not easily permit analysis of how safety degrades when the disturbance is more easily characterized by its impact in a Barrier Function derivative. This limitation motivates Projection-to-State Safety (PSSf), which enables a characterization of safety in terms of a projected disturbance.

We refer to a continuously differentiable function $\bs{\Pi}:\R^n\to\R^k$ as a \textit{projection}, and denote $\mb{y}=\bs{\Pi}(\mb{x})$. Considering the system governed by \eqref{eqn:eomdist}, the associated projected system is governed by the dynamics:
\begin{equation}
\label{eqn:projectedgraddyn}
    \dot{\mb{y}} = \mathbf{D}_{\bs{\Pi}}(\mb{x})\left(\mb{f}(\mb{x})+\mb{g}(\mb{x})\mb{u}\right)+\mb{D}_{\bs{\Pi}}(\mb{x})\mb{d},
\end{equation}
where $\mb{D}_{\bs{\Pi}}:\R^n\to\R^{k\times n}$ denotes the Jacobian of $\bs{\Pi}$. As will be seen when quantifying the impact of model uncertainty and learning error in Section \ref{sec:learning}, if the disturbance can be partially characterized in terms of the state and input, we may rewrite the projected dynamics as:
\begin{equation}
\label{eqn:projecteddyn}
    \dot{\mb{y}} = \mb{f}_\mb{y}(\mb{x})+\mb{g}_{\mb{y}}(\mb{x})\mb{u}+\bs{\delta},
\end{equation}
where $\mb{f}_{\mb{y}}:\R^n\to\R^k$ and $\mb{g}_{\mb{y}}:\R^n\to\R^{k\times m}$ are Lipschitz continuous on $\R^n$, and  $\bs{\delta}\in\R^k$ is referred to as the \textit{projected disturbance}. We note it is not explicitly necessary that the relationships $\mb{f}_{\mb{y}}(\mb{x}) = \mathbf{D}_{\bs{\Pi}}(\mb{x})\mb{f}(\mb{x})$, $\mb{g}_{\mb{y}}(\mb{x}) = \mathbf{D}_{\bs{\Pi}}(\mb{x})\mb{g}(\mb{x})$, and $\bs{\delta}=\mathbf{D}_{\bs{\Pi}}(\mb{x})\mb{d}$ hold, but are one possible relationship between the terms in \eqref{eqn:projectedgraddyn} and \eqref{eqn:projecteddyn}. 
For the following results, we will assume that $\bs{\delta}$ is essentially bounded in time and define $\Vert\bs{\delta}\Vert_\infty  \triangleq\esssup_{t\geq 0}\Vert\bs{\delta}(t)\Vert$. We are interested in relating behaviors of the projected system to the original system, motivating the following definition:

\begin{definition}[\textit{Projection-to-State Safety}]
The closed-loop system \eqref{eqn:cloopdist} is \textit{projection-to-state safe} (PSSf) on $\C$ with respect to the projection $\bs{\Pi}$ and projected disturbances $\bs{\delta}$ if there exists $\overline{\delta}>0$ and $\gamma\in\K_\infty$ such that the set $\C_{\bs{\delta}}\supset\C$,
\begin{align}
    \C_{\bs{\delta}} &\triangleq \left\{\mb{x} \in \R^n : h(\mb{x})+\gamma(\Vert\bs{\delta}\Vert_\infty) \geq 0\right\}, \label{eqn:safesetpss}\\
    \partial\C_{\bs{\delta}} &\triangleq \{\mb{x} \in \R^n : h(\mb{x})  +\gamma(\Vert\bs{\delta}\Vert_\infty) = 0\},\label{eqn:safesetboundarypss}\\
    \textrm{Int}(\C_{\bs{\delta}}) &\triangleq \{\mb{x} \in \R^n : h(\mb{x})+\gamma(\Vert\bs{\delta}\Vert_\infty) > 0\},\label{eqn:safetsetinteriorpss}
\end{align}
is forward invariant for all $\bs{\delta}$ satisfying $\Vert\bs{\delta}\Vert_\infty\leq\overline{\delta}$.
\end{definition}
In contrast to the definition of ISSf which enlarges the safe set in terms of the disturbance $\mb{d}$, PSSf quantifies how the safe set enlarges in terms of the projected disturbance $\bs{\delta}$. To utilize safety guarantees implied by ISSf-CBFs for analyzing PSSf behavior, we require the following definition:

\begin{definition}[\textit{Compatible Projection}]
A function $h_{\bs{\Pi}}:\R^k\to\R$ is said to be a \textit{compatible projection} for the function $h:\R^n\to\R$ with respect to the projection $\bs{\Pi}:\R^n\to\R^k$ if there exists $\underline{\sigma}, \overline{\sigma}\in\cal{K}_{\infty,e}$ such that for all $\mb{x}\in\R^n$:
\begin{equation}
\label{eqn:compatproj}
    \underline{\sigma}(h(\mb{x})) \leq h_{\bs{\Pi}}(\bs{\Pi}(\mb{x}))\leq\overline{\sigma}(h(\mb{x})).
\end{equation}
\end{definition}

\begin{remark}
If $h$ and $h_{\bs{\Pi}}$ are norms on $\R^n$ and $\R^k$, respectively, then $\bs{\Pi}$ reduces to a \textit{dynamic projection} as introduced in \cite{taylor2019control}. Whereas dynamic projections preserve the topological notion of a point between the state and projected spaces, compatible projections can preserve more interesting topological structures such as sets.
\end{remark}

\begin{remark}
The definition of a compatible projection can be abstractly viewed through the lens of category theory, mirroring the idea that one proves a property by mapping a system to the ``simplist'' type of system that has that property \cite{ames2006stability}.  For safety, these are dynamical systems defined on the entire real line, with the safe set being the positive reals.  Thus $h_{\bs{\Pi}}$ is a compatible projection if the following diagram: 
\begin{equation*}
\begin{tikzcd}
\R^n \arrow[d,"\bs{\Pi}"'] \arrow[r, "h"] &  \R \\
\R^k \arrow[ru,"h_{\bs{\Pi}}"'] &  
\end{tikzcd}
\end{equation*}
\emph{commutes up to class $\cal{K}$ functions}, i.e., \eqref{eqn:compatproj} being satisfied.
\end{remark}

In the context of safety, if a set $\C\subset\R^n$ is defined via a continuously differentiable function $h$ as in \eqref{eqn:safeset}-\eqref{eqn:safetsetinterior}, a compatible projection $h_{\bs{\Pi}}$ for the function $h$ with respect to $\bs{\Pi}$ defines a corresponding set $\C_{\bs{\Pi}}\subset\R^k$:
\begin{align}
    \C_{\bs{\Pi}} &\triangleq \left\{\mb{y} \in \R^k : h_{\bs{\Pi}}(\mb{y}) \geq 0\right\}, \label{eqn:safesetproj}\\
    \partial\C_{\bs{\Pi}} &\triangleq \{\mb{y} \in \R^k : h_{\bs{\Pi}}(\mb{y}) = 0\},\label{eqn:safesetboundaryproj}\\
    \textrm{Int}(\C_{\bs{\Pi}}) &\triangleq \{\mb{y} \in \R^k : h_{\bs{\Pi}}(\mb{y}) > 0\}.\label{eqn:safetsetinteriorproj}
\end{align}
The inequalities in \eqref{eqn:compatproj} preserve the notion of what states are considered safe between the state space and projected space, such that $\mb{x}\in\C\implies\bs{\Pi}(\mb{x})\in\C_{\bs{\Pi}}$. The preceding implication is also true of the boundaries and interiors of the two sets. The following theorem allows us to extend ISSf properties of the projected system on $\C_{\bs{\Pi}}$ to PSSf properties of the original system on $\C$.

\begin{theorem}
\label{thm:isscbf2pssf}
Let $\C\subset\R^n$ be the 0-superlevel set of a continuously differentiable function $h:\R^n\to\R$ with $0$ a regular value. The disturbed system \eqref{eqn:eomdist} can be rendered PSSf on $\C$ with respect to the projection $\bs{\Pi}$ and projected disturbances $\bs{\delta}$ if there exists a compatible projection $h_{\bs{\Pi}}$ for $h$ with respect to $\bs{\Pi}$ and Lipschitz continuous controller $\mb{k}:\R^n\to\R^m$ such that $h_{\bs{\Pi}}$ is an ISSf-CBF for the projected dynamics \eqref{eqn:projecteddyn} on $\C_{\bs{\Pi}}$ and $\mb{k}(\mb{x})\in K_{\textrm{issf}}(\mb{x})$ with:
\begin{equation*}
    K_{\textrm{issf}}(\bs{\mb{x}}) \triangleq \left\{\mb{u}\in\R^m ~\left|~ \begin{array}{l} \dot{h}_{\bs{\Pi}}(\bs{\Pi}(\mb{x}),\mb{u})\geq\\-\alpha(h_{\bs{\Pi}}(\bs{\Pi}(\mb{x})))-\iota(\Vert\bs{\delta}\Vert) \end{array}\right. \right\},
\end{equation*}
\end{theorem}
\begin{proof}
As $h_{\bs{\Pi}}$ is an ISSf-CBF for \eqref{eqn:projecteddyn} on $\C_{\bs{\Pi}}$ and the state-feedback controller satisfies $\mb{k}(\mb{x})\in K_{\textrm{issf}}(\mb{x})$, Theorem \ref{thm:issfcbf2issf} implies that the controller $\mb{k}$ renders the set $\C_{\bs{\Pi}}$ input-to-state safe for all $\bs{\delta}$ satisfying $\Vert\bs{\delta}\Vert_\infty\leq\overline{\delta}$. In particular, there exists $\gamma\in\cal{K}_\infty$ such that the set:
\begin{equation}
    \C_{\bs{\Pi},\bs{\delta}} \triangleq \left\{\mb{y}\in\R^k ~|~ h_{\bs{\Pi}}(\mb{y})+\gamma(\Vert\bs{\delta}\Vert_\infty)\geq 0\right\},
\end{equation}
is safe. Let $\mb{x}_0\in\R^n$ be such that $\mb{y}_0=\bs{\Pi}(\mb{x}_0)\in\C_{\bs{\Pi},\bs{\delta}}$. With $\mb{x}(0)=\mb{x}_0$ (implying $\mb{y}(0)=\mb{y}_0$), safety of $\C_{\bs{\Pi},\bs{\delta}}$ implies:
\begin{equation}
    h_{\bs{\Pi}}(\bs{\Pi}(\mb{x}(t))) + \gamma(\Vert\bs{\delta}\Vert_{\infty})\geq 0,
\end{equation}
for $t\in I(\mb{x}_0)$. As $h_{\bs{\Pi}}$ is a compatible projection for $h$ with respect to $\bs{\Pi}$, we have:
\begin{equation}
   \overline{\sigma}(h(\mb{x}(t))) + \gamma(\Vert\bs{\delta}\Vert_{\infty})\geq 0,
\end{equation}
Multiplying both sides by $\frac{1}{2}$ and using that $\overline{\sigma}\in\cal{K}_{\infty,e}$, it follows that:
\begin{equation}
   \overline{\sigma}^{-1}\left(\frac{1}{2}\overline{\sigma}(h(\mb{x}(t))) + \frac{1}{2}\gamma(\Vert\bs{\delta}\Vert_{\infty})\right)\geq 0,
\end{equation}
The triangle inequality for class $\cal{K}$ functions \cite{kellett2014compendium} implies:
\begin{equation}
    h(\mb{x}(t))+\underbrace{\overline{\sigma}^{-1}(\gamma(\Vert\bs{\delta}\Vert_\infty))}_{\gamma'(\Vert\bs{\delta}\Vert_\infty)} \geq 0,
\end{equation}
for all $t\in I(\mb{x}_0)$, implying the set $\C_{\bs{\delta}}$ defined as in \eqref{eqn:safesetpss}-\eqref{eqn:safetsetinteriorpss} using $\gamma'$ is forward invariant, and hence safe. Thus the closed-loop system \eqref{eqn:eomdist} is PSSf on $\C$ with respect to $\bs{\Pi}$ and corresponding projected disturbances $\bs{\delta}$.
\end{proof}


\begin{corollary}
\label{cor:cbf2pssf}
Let $\C\subset\R^n$ be the 0-superlevel set of a continuously differentiable function $h:\R^n\to\R$ with $0$ a regular value. 
Viewing $h$ as a projection such that $y=h(\mb{x})$, let the projected dynamics be given by:
\begin{equation}
\label{eqn:projscalardyn}
\dot{y} = f_y(\mb{x})+\mb{g}_y(\mb{x})\mb{u}+\delta
\end{equation}
with projected disturbances $\delta\in\R$. If there exists a Lipschitz continuous feedback controller $\mb{k}:\R^n\to\R^m$ such that:
\begin{equation}
\label{eqn:projcbfcond}
   f_y(\mb{x}) + \mb{g}_y(\mb{x})\mb{k}(\mb{x}) \geq -\alpha(y),
\end{equation}
and there exists $\overline{\delta}>0$ satisfying $\vert\delta\vert<\overline{\delta}$, then the disturbed system \eqref{eqn:eomdist} can be rendered PSSf on $\C$ with respect to the projection $h$ and projected disturbances $\delta$.
\end{corollary}
\begin{proof}
We first note that the identity map $I:\R\to\R$ is a compatible projection for $h$:
\begin{equation}
    h(\mb{x}) \leq I(h(\mb{x})) \leq h(\mb{x})
\end{equation}
with $\underline{\sigma}(r)=\overline{\sigma}(r) = r$. Furthermore, the inequality in \eqref{eqn:projcbfcond} implies the identity map can be viewed as an ISSf-CBF for the projected dynamics \eqref{eqn:projscalardyn}:
\begin{equation}
    \sup_{\mb{u}\in\R^m}\dot{I}(\mb{x},\mb{u},\delta) \geq\dot{I}(\mb{x},\mb{k}(\mb{x}),\delta)\geq -\alpha(I(y))-\vert\delta\vert,
\end{equation}
for all $\mb{x}\in\R^n$ and $\delta\in\R$ satisfying $\vert\delta\vert\leq\overline{\delta}$. Therefore the system \eqref{eqn:eomdist} can be rendered PSSf on $\C$ with respect to the projection $h$ and projected disturbances $\delta$ by Theorem \ref{thm:isscbf2pssf}. 
\end{proof}

\section{Integration With Learning}
\label{sec:learning}
In this section we consider a structured form of uncertainty in affine control systems. We discuss the impact of this uncertainty in a CBF time derivative, and on the PSSf behavior of the system. We demonstrate how learning can be used to mitigate the resulting impact on safety.

In practice, the system dynamics \eqref{eqn:eom} are not known during  control design  due to parametric error and unmodeled dynamics. Instead, a nominal model of the system is utilized:
\begin{equation}
\label{eqn:eomnominal}
    \widehat{\dot{\mb{x}}} = \widehat{\mb{f}}(\mb{x}) + \widehat{\mb{g}}(\mb{x})\mb{u},
\end{equation}
where $\widehat{\mb{f}}:\R^n\to\R^n$ and $\widehat{\mb{g}}:\R^n\to\R^{n\times m}$ are assumed to be Lipschitz continuous on $\R^n$. By adding and subtracting \eqref{eqn:eomnominal} to \eqref{eqn:eom}, the dynamics of the system can be expressed as:
\begin{equation}
    \dot{\mb{x}} = \widehat{\mb{f}}(\mb{x}) + \widehat{\mb{g}}(\mb{x})\mb{u} + \overbrace{\underbrace{\mb{f}(\mb{x})-\widehat{\mb{f}}(\mb{x})}_{\mb{b}(\mb{x})}+\underbrace{\left(\mb{g}(\mb{x})-\widehat{\mb{g}}(\mb{x})\right)}_{\mb{A}(\mb{x})}\mb{u}}^{\mb{d}},
\end{equation}
where the unknown disturbance $\mb{d}=\mb{b}(\mb{x})+\mb{A}(\mb{x})\mb{u}$ is assumed to be time invariant, but explicitly depends on the state and input to the system.

\begin{figure*}[h]
    \hspace*{0 cm}
     \centering
    \begin{subfloat}
        {\includegraphics[trim={0in, 0in, 0in, 0in}, clip, scale =0.172, valign =t ]{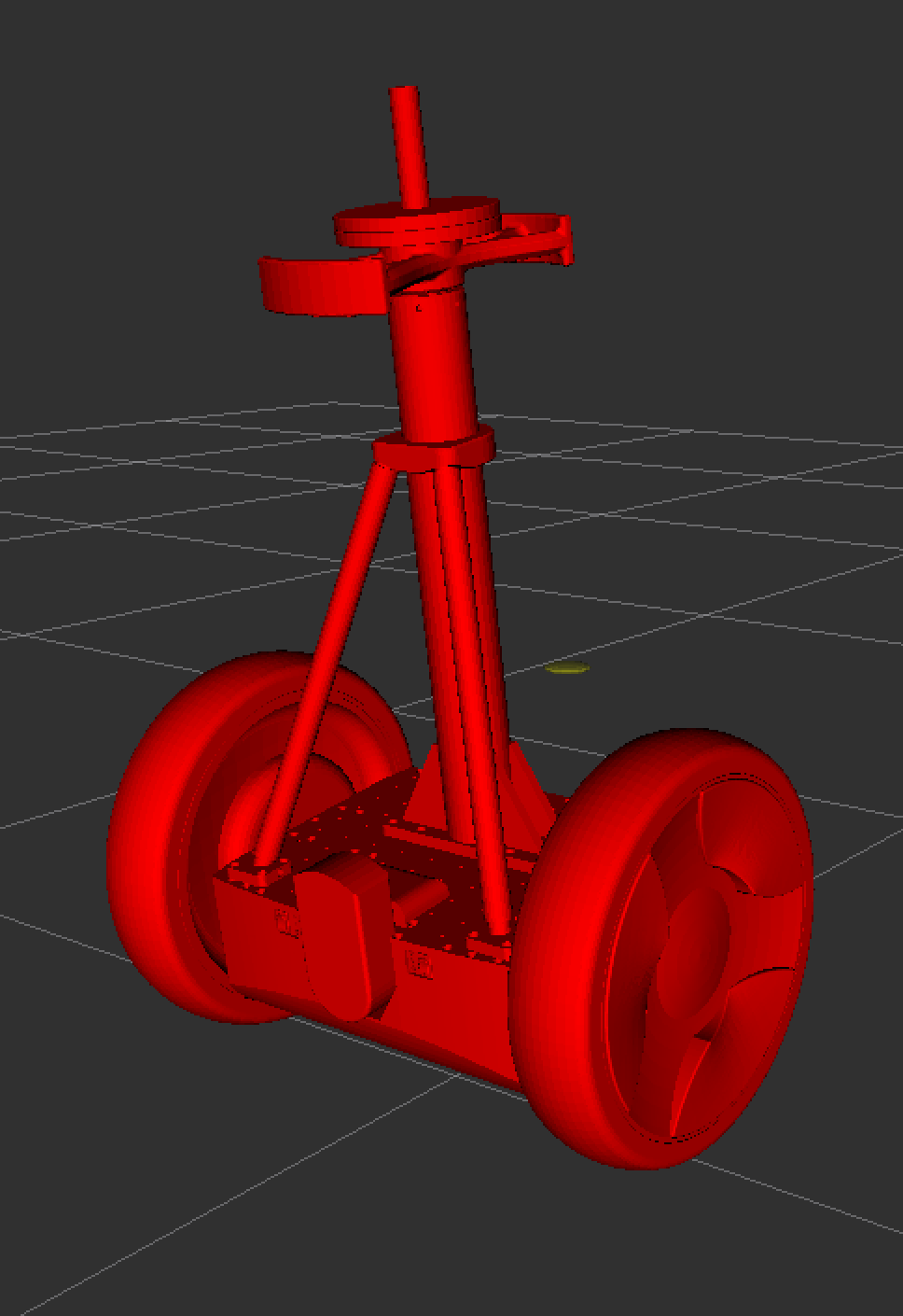}}
    \end{subfloat}
    \hfill
    \hspace*{0 cm}
    \begin{subfloat}
       {\includegraphics[scale =0.43, valign =t]{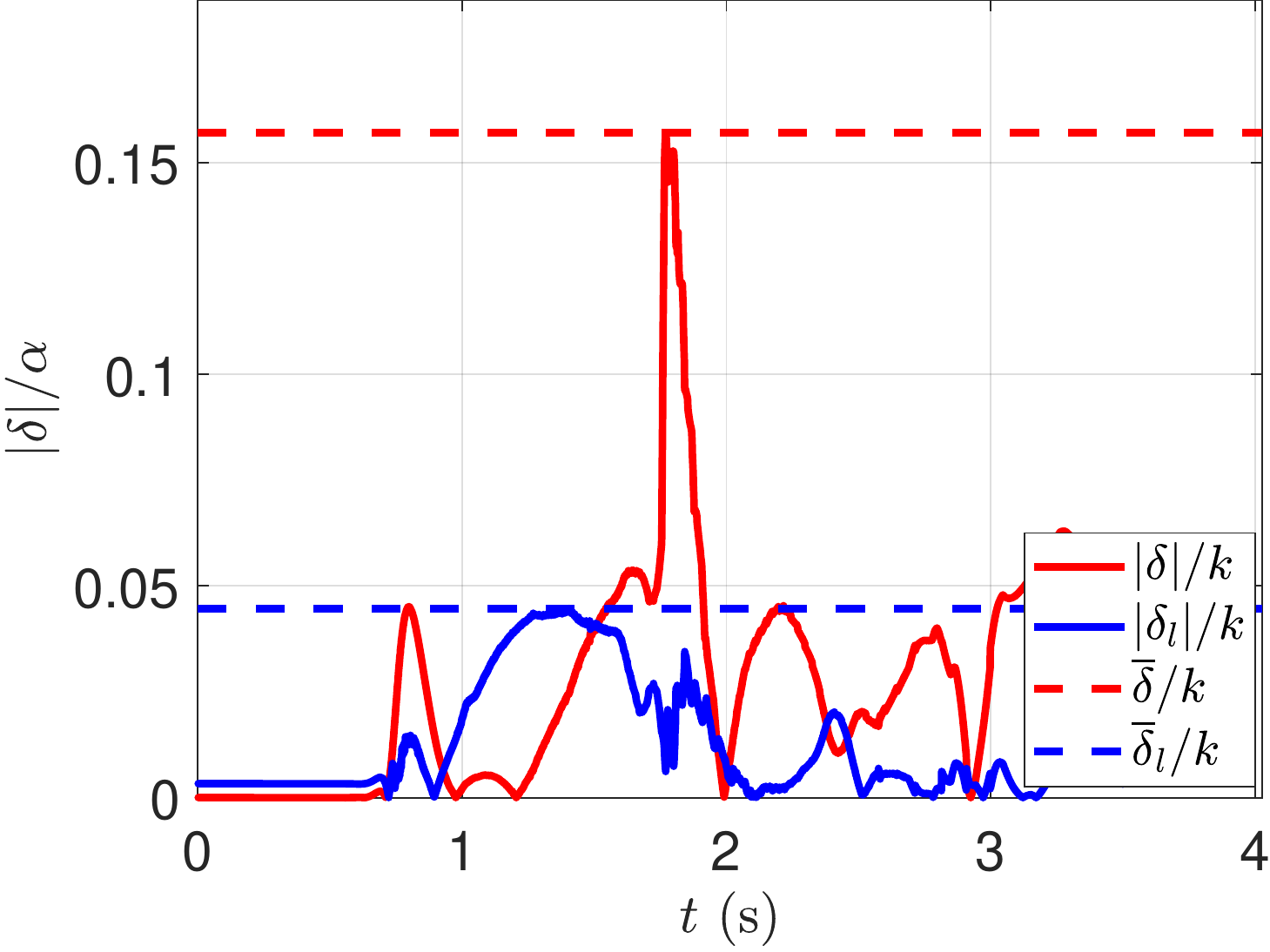}}
    \end{subfloat}
    \hfill
    \begin{subfloat}
     {\includegraphics[scale =0.43,valign=t]{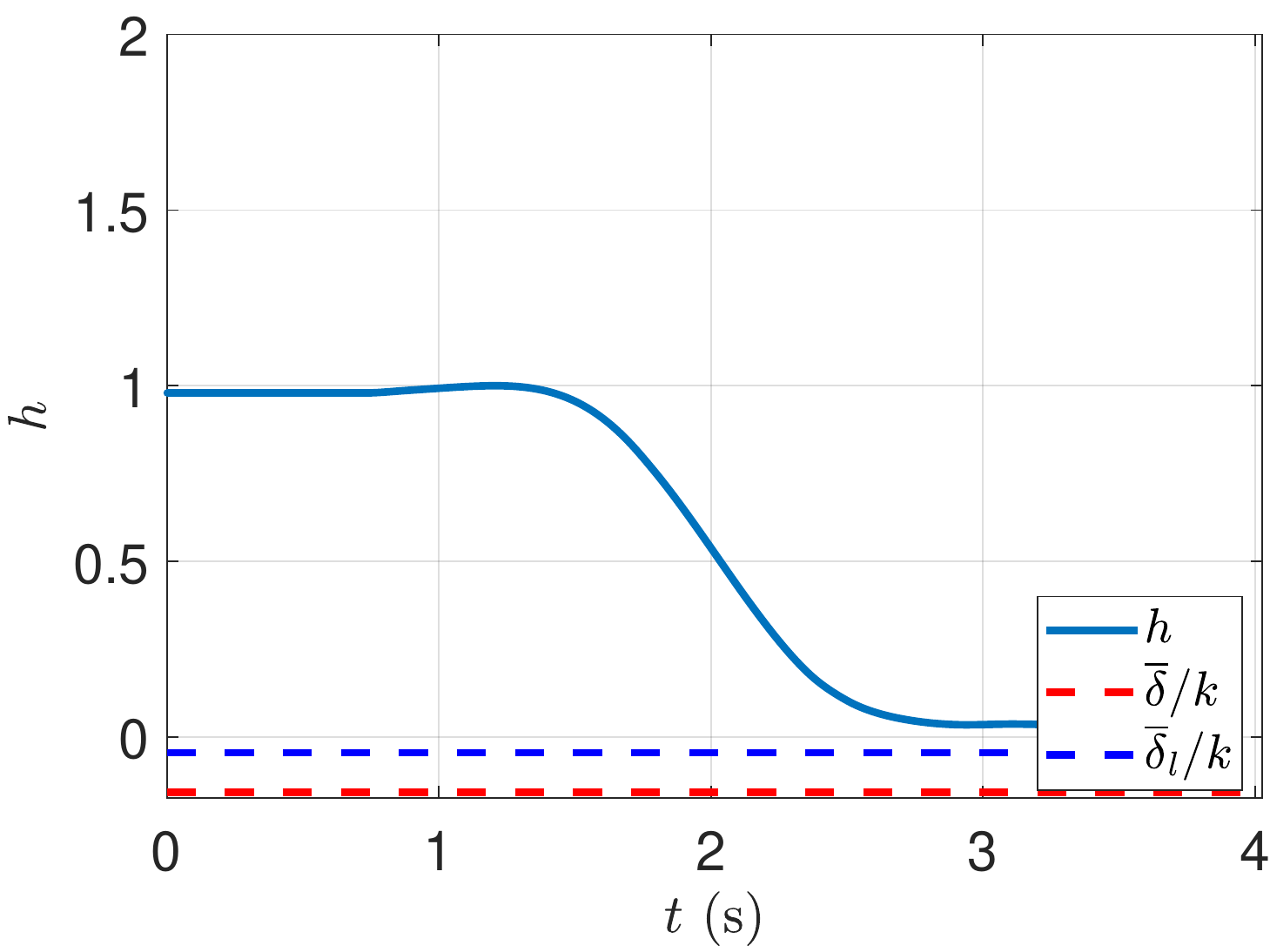}}
    \end{subfloat}
        \caption{Simulation results with Segway platform demonstrating improvement in PSSf behavior. \textbf{(Left)} Robotic Segway platform model used in simulation. \textbf{(Center)} Absolute value of the projected disturbance $\delta$ along the trajectory without learning models (\eqref{eqn:projdist},red) and with learning models (\eqref{eqn:projdistest}, blue), with learning reducing the worse case projected disturbance ($\overline{\delta}/\alpha$). 
        \textbf{(Right)} The value of the barrier satisfies the corresponding worst case lower bound with and without learning being used to compute $\delta$. The worst case lower bound is raised with learning (the blue dashed line lies above the red dashed line).
        }
    \label{fig:simresults}
\end{figure*}

If the function $h:\R^n\to\R$ is a CBF for the nominal model \eqref{eqn:eomnominal} on $\C$, the uncertainty in the dynamics directly manifests in the time derivative of $h$:
\begin{align}
\label{eqn:hdotdecomp}
    \dot{h}(\mb{x},\mb{u})  = & \overbrace{\derp{h}{\mb{x}}(\mb{x})(\widehat{\mb{f}}(\mb{x})+\widehat{\mb{g}}(\mb{x})\mb{u})}^{\widehat{\dot{h}}(\mb{x},\mb{u})} \nonumber \\ & + \underbrace{\derp{h}{\mb{x}}(\mb{x})\mb{b}(\mb{x})}_{b(\mb{x})} + \underbrace{\derp{h}{\mb{x}}(\mb{x})\mb{A}(\mb{x})}_{\mb{a}(\mb{x})^\top}\mb{u}.
\end{align}
Given that $h$ is a CBF for \eqref{eqn:eomnominal} on $\C$, let $\mb{k}:\R^n\to\R^m$ be a Lipschitz continuous state-feedback controller such that:
\begin{equation}
\label{eqn:hdotassump}
   \sup_{\mb{u}\in\R^m} \widehat{\dot{h}}(\mb{x},\mb{u}) \geq  \widehat{\dot{h}}(\mb{x},\mb{k}(\mb{x})) \geq -\alpha(h(\mb{x})).
\end{equation}
Letting the projected disturbance be defined as: 
\begin{equation}
\label{eqn:projdist}
\delta = b(\mb{x})+\mb{a}(\mb{x})^\top\mb{k}(\mb{x}),
\end{equation}
Corollary \ref{cor:cbf2pssf} implies that if there exists a $\overline{\delta}>0$ such that $\vert b(\mb{x})+\mb{a}(\mb{x})^\top\mb{k}(\mb{x})\vert\leq\overline{\delta}$ for all $\mb{x}\in\R^n$, the uncertain system \eqref{eqn:eom} can be rendered PSSf on $\C$ with respect to the projection $h$ and projected disturbances $\delta$.

As in \cite{taylor2019learning}, we may wish to reduce the error between $\dot{h}$ and $\widehat{\dot{h}}$ by utilizing data-driven models to estimate the functions $b$ and $\mb{a}$. In particular, given Lipschitz continuous estimators $\widehat{b}:\R^n\to\R$ and $\widehat{\mb{a}}:\R^n\to\R^m$, \eqref{eqn:hdotdecomp} can be reformulated as:
\begin{align}
\label{eqn:hdotdecompestimators}
    \dot{h}(&\mb{x},\mb{u})  =  \overbrace{\derp{h}{\mb{x}}(\mb{x})(\widehat{\mb{f}}(\mb{x})+\widehat{\mb{g}}(\mb{x})\mb{u})+\widehat{b}(\mb{x})+\widehat{\mb{a}}(\mb{x})^\top\mb{u}}^{\widehat{\dot{h}}(\mb{x},\mb{u})} \nonumber \\ & + \underbrace{\derp{h}{\mb{x}}(\mb{x})\mb{b}(\mb{x})-\widehat{b}(\mb{x})}_{\tilde{b}(\mb{x})} + \underbrace{\left(\derp{h}{\mb{x}}(\mb{x})\mb{A}(\mb{x})-\widehat{\mb{a}}(\mb{x})^\top\right)}_{\tilde{\mb{a}}(\mb{x})^\top}\mb{u}.
\end{align}
Under the assumption that the introduction of the estimators does not violate the CBF condition, such that there exists a state-feedback controller $\mb{k}$ satisfying \eqref{eqn:hdotassump} with $\widehat{\dot{h}}$ defined as in \eqref{eqn:hdotdecompestimators}, we may define the projected disturbance as:
\begin{equation}
\label{eqn:projdistest}
    \delta = \tilde{b}(\mb{x}) + \tilde{\mb{a}}(\mb{x})^\top\mb{k}(\mb{x})
\end{equation}
As before, if there exists $\overline{\delta}>0$ such that $\vert\tilde{b}(\mb{x}) + \tilde{\mb{a}}(\mb{x})^\top\mb{k}(\mb{x})\vert\leq\overline{\delta}$ for all $\mb{x}\in\R^n$, Corollary \ref{cor:cbf2pssf} can be used to certify \eqref{eqn:eom} as PSSf on $\C$ with respect to the projection $h$ and projected disturbances $\delta$. The preceding statements are formalized in the following theorem:
\begin{theorem}
Let $\C\subset\R^n$ be the 0-superlevel set of a continuously differentiable function $h:\R^n\to\R$ with $0$ a regular value, and let $\widehat{h}:\R^n\to\R$ be defined as in \eqref{eqn:hdotdecomp} or \eqref{eqn:hdotdecompestimators}. If there exists a Lipschitz continuous state-feedback controller $\mb{k}:\R^n\to\R^m$ satisfying \eqref{eqn:hdotassump}, and $\overline{\delta}>0$ such that the corresponding projected disturbance defined as in \eqref{eqn:projdist} or \eqref{eqn:projdistest} satisfies $\vert\delta\vert\leq\overline{\delta}$, then \eqref{eqn:eom} is PSSf on $\C$ with respect to the projection $h$ and projected disturbances $\delta$.
\end{theorem}

\begin{figure*}[h]
    \hspace*{0 cm}
     \centering
    \begin{subfloat}
        {\includegraphics[trim={8.5in, .1in, 7in, 1in}, clip, scale =0.1075, valign =t ]{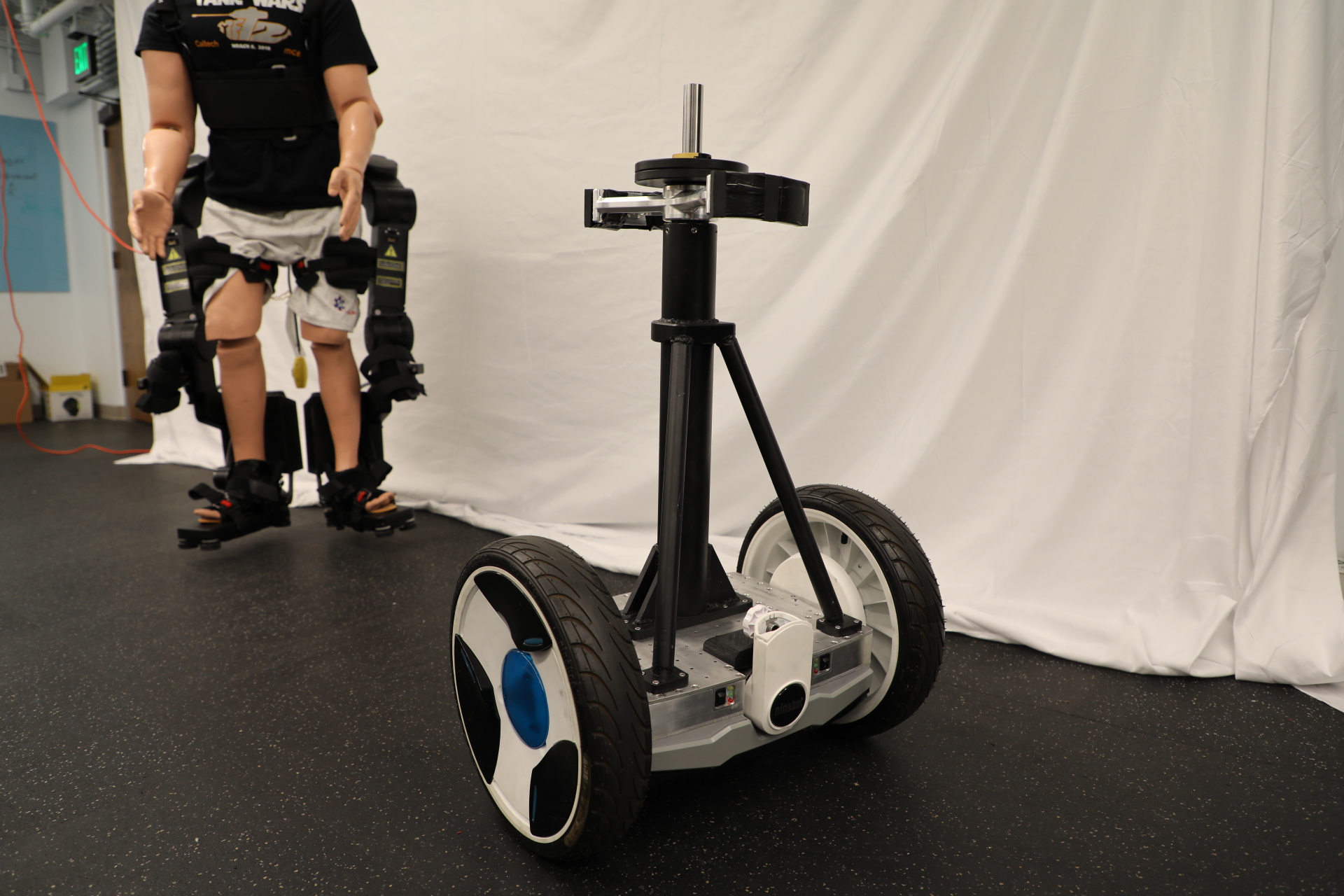}}
    \end{subfloat}
    \hfill
    \hspace*{0 cm}
    \begin{subfloat}
       {\includegraphics[scale =0.43, valign =t]{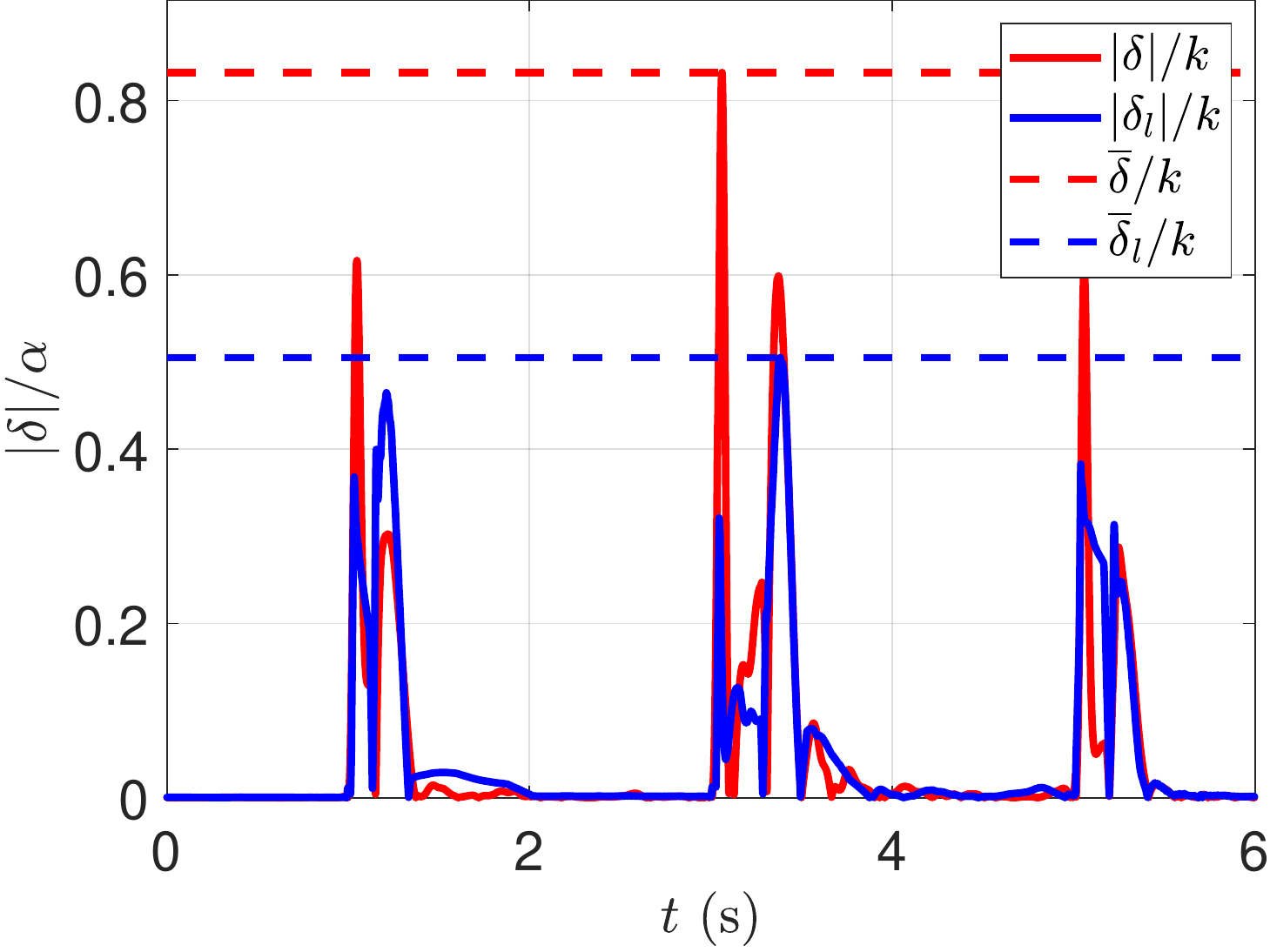}}
    \end{subfloat}
    \hfill
    \begin{subfloat}
     {\includegraphics[scale =0.43,valign=t]{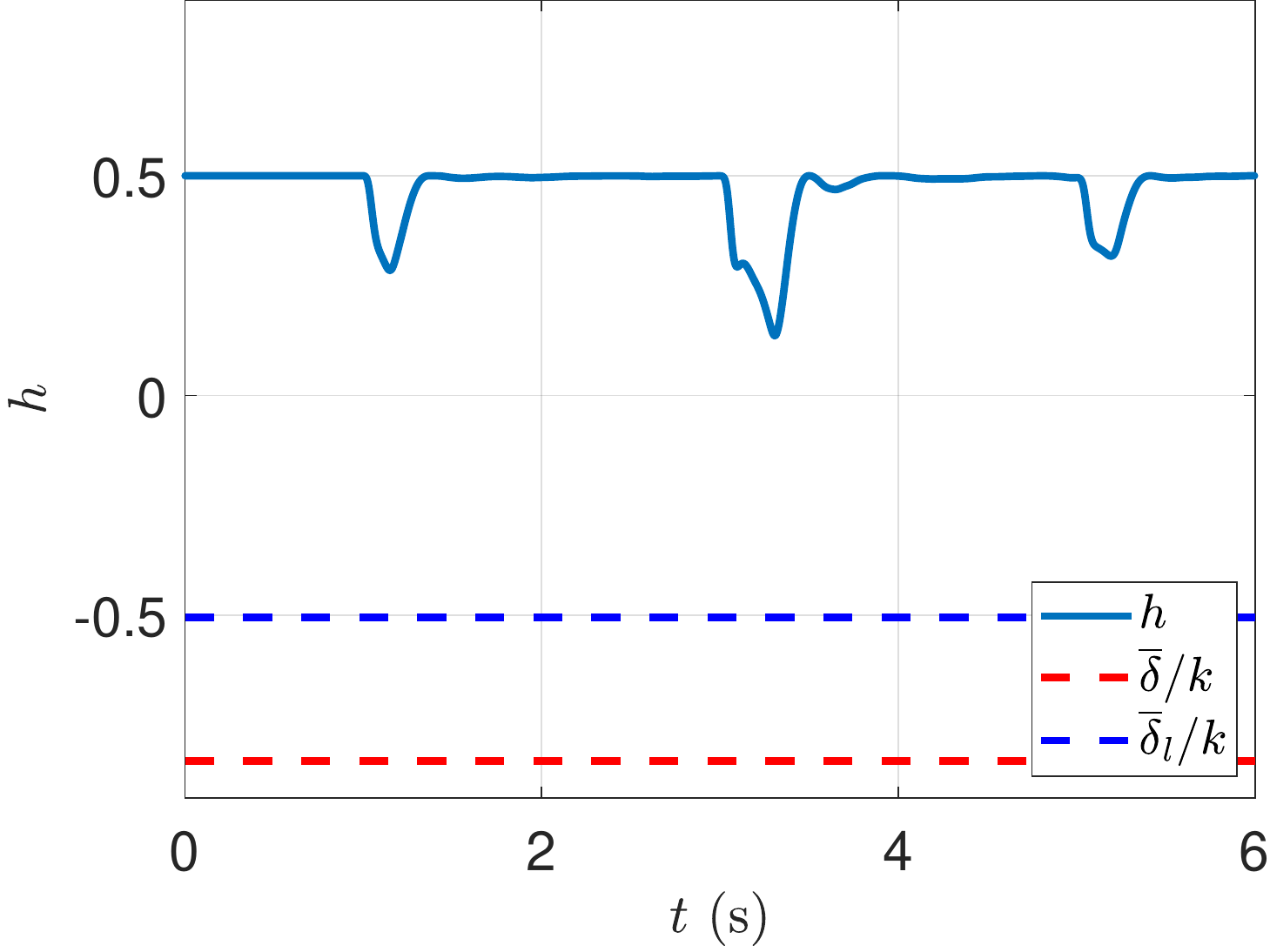}}
    \end{subfloat}
        \caption{Experimental results with Segway platform demonstrating improvement in PSSf behavior. \textbf{(Left)} Physical robotic Segway platform used in experimentation. \textbf{(Center)} Absolute value of the projected disturbance $\delta$ along the trajectory without learning models (\eqref{eqn:projdist},red) and with learning models (\eqref{eqn:projdistest}, blue), with learning reducing the worse case projected disturbance ($\overline{\delta}/\alpha$). 
        \textbf{(Right)} The value of the barrier satisfies the corresponding worst case lower bound with and without learning being used to compute $\delta$. The worst case lower bound is raised with learning (the blue dashed line lies above the red dashed line).
        }
    \label{fig:hwresults}
\end{figure*}

In the presence of estimators, this theorem defines a quantitative relationship between the prediction error of the estimators, $\vert\dot{h}(\mb{x},\mb{k}(\mb{x}))-\widehat{\dot{h}}(\mb{x},\mb{k}(\mb{x}))\vert= \vert\delta\vert$, and the degradation of the safety of the closed-loop system. As the prediction error is reduced (via additional training data or more complex learning models), the set kept safe more closely resembles $\C$.

\section{Experimental Validation}
\label{sec:simexp}
To demonstrate the ability of learning to improve safety guarantees via Projection-to-State Safety, we deployed the episodic learning framework with CBFs established in \cite{taylor2019learning} on a robotic Segway platform,  seen in Figure \ref{fig:simresults} and \ref{fig:hwresults}, in simulation and experimentally. The planar, 4 dimensional, Segway was considered, with states given by horizontal position, horizontal velocity, pitch angle, and pitch angle rate. The input the system is specified as a torque about the wheel at the base of the Segway. In both cases a sequence of episodes were ran to train estimators $\widehat{b}$ and $\widehat{\mb{a}}$. 

In each episode the Segway was set to track a desired trajectory in the pitch angle space without violating a barrier function on a portion of its state, using the safety-critical control formulation in \cite{gurriet2018towards}. After the sequence of episodes, the Segway was ran once more with a learning-informed controller, and the projected disturbance $\delta$ as defined in \eqref{eqn:projdist} and \eqref{eqn:projdistest} was computed. The worst case disturbance $\overline{\delta}$ was found, and a lower bound on $h$ for that trajectory was determined using the fact $h\leq\alpha^{-1}(\overline{\delta})\implies\dot{h}\geq 0$. In both simulation and experimental results, $\alpha(r) = kr$ with $k>0$. 

In simulation, the Segway was given a bound on its position in space, forcing it to stay within one meter of its starting location. The CBF was generated through the backup controller method \cite{gurriet2018online}. The value of the CBF is computed at each time-step by integrating the system forward in time under a backup control law. Sensitivity analysis along the trajectory is used to compute the gradient of the CBF. This simulation result highlights the ability of learning to reduce worst case disturbances for complex CBFs that cannot be expressed in closed-form. The simulation was done in a ROS-based C++ environment \cite{singletary2019}. The simulation environment accurately simulates the physical system by adding input delay, sensor noise, and state estimation. Experimentally, a simple CBF was specified to limit the pitch angle and pitch angle rate of the Segway to an ellipse about the Segway's equilibrium state. The desired pitch angle trajectory would lead to the Segway tipping quickly, thereby violating the safety set in the absence of the CBF and safety-critical control formulation.

In both cases, we see that introducing learning estimators into the computation of the projected disturbance decreases the worse case disturbance ($\overline{\delta}>\overline{\delta}_l$). This leads to a greater lower bound on $h$, and thus a stronger guarantee on the PSSf behavior of the system. We note that the conservative nature of the lower bounds on $h$ arise from the fact that the worst case disturbance $\overline{\delta}$ along the trajectory is used. If the worst case disturbance can be reduced (by data-aware control synthesis), stronger guarantees on safety can be made.

\section{Conclusions}
We presented a novel method for assessing the impact of disturbances on safety in a project environment via Projection-to-State Safety, and considered how it can be utilized in conjunction with learning to mitigate the impact of model uncertainty on safety. We demonstrate the ability of learning to improve the guarantees endowed by PSSf in simulation and experimentally on a Segway platform. Future work includes developing data-driven methods for quantifying the worst case projected disturbance, and synthesizing data-aware controllers that reduce the projected disturbance.

\bibliographystyle{plain}
\bibliography{taylor_main}
\end{document}